%% file: efunction.tex
\documentclass{article}
\usepackage{amsmath,amssymb}
\usepackage{amsthm}
\usepackage{subfigure}
\usepackage{tikz}
\usepackage[T1]{fontenc}
\usepackage[ruled,vlined,linesnumbered,english]{algorithm2e}

\newtheorem{lema}{Lemma}
\newtheorem{teorema}[lema]{Theorem}
\newtheorem{corolario}[lema]{Corollary}
\newtheorem{conjectura}[lema]{Conjecture}

\begin{document}

\title{An energy function and its application to the periodic behavior of $k$-reversible processes}

\author{Leonardo~I.~L.~Oliveira\\
Valmir~C.~Barbosa\\
\\
Programa de Engenharia de Sistemas e Computa\c c\~ao, COPPE\\
Universidade Federal do Rio de Janeiro\\
Caixa Postal 68511, 21941-972 Rio de Janeiro - RJ, Brazil\\
\\
F\'abio~Protti\thanks{Corresponding author (fabio@ic.uff.br).}\\
\\
Instituto de Computa\c c\~ao\\
Universidade Federal Fluminense\\
Rua Passo da P\'atria, 156, 24210-240 Niter\'oi - RJ, Brazil
}

\date{}

\maketitle

\begin{abstract}
We consider the graph dynamical systems known as $k$-reversible processes. 
In such processes, each vertex in the graph has one of two possible
states at each discrete time step. Each vertex changes its state between the current
time and the next if and only if it currently has at least $k$ neighbors in a
state different than its own. For such processes, we present a monotonic function similar to the decreasing energy functions
used to study threshold networks. Using this new function, we show an alternative proof for the maximum period length in a $k$-reversible process and provide
better upper bounds on the transient length in both the general case and the case of trees.

\end{abstract}

\newpage
\section{Introduction}

Let $G$ be a simple, undirected, finite graph with $n$ vertices and $m$ edges.
The set of vertices of $G$ is denoted by $V(G) = \{v_1, v_2, \ldots, v_n\}$, its
set of edges by $E(G)$, and its maximum node degree by $\Delta(G)$. A \textit{$k$-reversible} process on $G$
is an iterative process in which, at each discrete time $t$, each vertex in $G$
has one of two possible states. The state of a vertex is represented by an integer
belonging to the set $\{-1,+1\}$ and each vertex has its state changed from
one time to the next if and only if it currently has at least $k$ neighbors in a
state different than its own, where $k$ is a positive integer.

We denote by $x(t)$ the sequence in which the $i$th component, $x_{i}(t)$, is the state of vertex $v_{i}$ at time $t$, and by $\mathit{op}(t)$ the sequence 
in which the $i$th component, $\mathit{op}_{i}(t)$, is the number of neighbors of $v_{i}$ in a state different than its own at time $t$. Any sequence $x(t)$ is called 
a \textit{configuration} and the sequence $x(0)$ is an \textit{initial configuration}. Regarding the periodic behavior of such processes, $p(x(0))$ denotes the length of
the \textit{period} reached after a finite number of steps starting with configuration $x(0)$. The length of the
\textit{transient} is denoted by $\tau(x(0))$. Formally, $p(x(0))$ and $\tau(x(0))$ are integer numbers such that:
\begin{itemize}
\item $x(t+p(x(0))) = x(t)$ for any $t \geq \tau(x(0))$;
\item $x(t+q) \neq x(t)$ for any $t < \tau(x(0))$ or $q < p(x(0))$.
\end{itemize}

The motivation to study $k$-reversible processes is related to the analysis of opinion dissemination in social networks. As graph dynamical systems, their study
is multidisciplinary and related to several other areas, like optics \cite{phase-wrapping}, neural networks \cite{hopfield}, statistical mechanics \cite{bootstrap-example}, and 
opinion \cite{opinion} and disease \cite{disease-example} dissemination. In the work by Dreyer \cite{processosReversiveis}, some important results regarding the periodic behavior of 
$k$-reversible processes are presented; for instance, $\tau(x(0))$ is $O(m + n^2)$ and $p(x(0)) \leq 2$.  Most of these results are based on reductions
from the so-called \textit{threshold networks}, which are broadly studied by Goles and Olivos \cite{golesOlivos, golesOlivos2}. 

It is known that, for all threshold networks, $p(x(0)) \leq 2$ \cite{energyFunction, poljak}. An intuitive approach to prove this result is based on a monotonic function 
called an \textit{energy function}. Its definition is very similar to that of the energy function associated with Hopfield networks \cite{hopfield} and it is a Lyapunov function. This function is used to prove 
several results associated with the period and transient lengths of threshold and majority networks.

The remainder of the paper is organized as follows. In Section~\ref{intr} we describe the new energy function for $k$-reversible processes and some of its properties. In Section~\ref{results} 
we use this function to provide an alternative proof of the maximum length of the period and of the transient length for $k$-reversible processes, for the general case and also for trees. 
Section~\ref{concl} contains our conclusions.

\section{\boldmath An energy function for $k$-reversible processes}\label{intr}

Let $S_1(t)$ and $S_2(t)$ be subsets of $V(G)$, defined as a function of time $t$ as follows:
\begin{equation}
S_{1}(t) = \{ v_i \mid \mathit{op}_i(t) \geq k \} \mbox{ and } S_{2}(t) = \{ v_i \mid \mathit{op}_i(t) < k \}.
\end{equation}
We define a nonnegative energy function $E(t)$ for $k$-reversible processes as follows:
\begin{equation} \label{eq:energia}
E(t) = \sum_{i \in S_{1}(t)}^{} (\mathit{op}_{i}(t) - k) + \sum_{i \in S_{2}(t)}^{} (k - \mathit{op}_{i}(t)). 
\end{equation}
We also define an auxiliary function $E'(t)$ that will be helpful later in this section to prove that $E(t)$ is a monotonically nondecreasing function:
\begin{equation} \label{eq:energia2}
E'(t) = \sum_{i \in S_{1}(t)}^{} (\mathit{op}_{i}(t+1) - k) + \sum_{i \in S_{2}(t)}^{} (k - \mathit{op}_{i}(t+1)). 
\end{equation}

\begin{lema}
\label{lema:lema1}
$E(t) = E'(t)$ for all $t \geq 0$.
\end{lema}

\begin{proof}
From (\ref{eq:energia}) and (\ref{eq:energia2}), we have
\begin{equation}
\label{eq:abreEt}
E(t) = \sum\limits_{i \in S_{1}(t)}^{} \mathit{op}_{i}(t) - \sum\limits_{i \in S_{2}(t)}^{} \mathit{op}_{i}(t) - |S_{1}(t)|k + |S_{2}(t)|k
\end{equation}
and
\begin{equation}
\label{eq:abreEt2}
E'(t) = \sum\limits_{i \in S_{1}(t)}^{} \mathit{op}_{i}(t+1) - \sum\limits_{i \in S_{2}(t)}^{} \mathit{op}_{i}(t+1) - |S_{1}(t)|k + |S_{2}(t)|k.
\end{equation}
Thus, we need to show that
\begin{equation}
\label{eqSets}
\sum\limits_{i \in S_{1}(t)}^{} \mathit{op}_{i}(t) - \sum\limits_{i \in S_{2}(t)}^{} \mathit{op}_{i}(t) = \sum\limits_{i \in S_{1}(t)}^{} \mathit{op}_{i}(t+1) - 
\sum\limits_{i \in S_{2}(t)}^{} \mathit{op}_{i}(t+1).
\end{equation}

First we define the following sets:
\begin{itemize}
\item $A(t) = \{ (v_i,v_j) \in E(G) \mid v_i \in S_{1}(t), v_j \in S_{1}(t), x_i(t) \neq x_j(t) \};$
\item $B(t) = \{ (v_i,v_j) \in E(G) \mid v_i \in S_{2}(t), v_j \in S_{2}(t), x_i(t) \neq x_j(t) \};$
\item $C(t) = \{ (v_i,v_j) \in E(G) \setminus (A(t) \cup B(t)) \mid x_i(t) \neq x_j(t)\}.$
\end{itemize}
Then we note that $\sum\limits_{i \in S_{1}(t)}^{} \mathit{op}_{i}(t) = 2|A(t)| + |C(t)|$  and that $\sum\limits_{i \in S_{2}(t)}^{} \mathit{op}_{i}(t) = 2|B(t)| + |C(t)|$. Therefore,
\begin{equation}
\label{eq:atbt}
\sum\limits_{i \in S_{1}(t)}^{} \mathit{op}_{i}(t) - \sum\limits_{i \in S_{2}(t)}^{} \mathit{op}_{i}(t) = 2|A(t)| - 2|B(t)|.
\end{equation}

Next we define the following sets:
\begin{itemize}
\item $A'(t) = \{ (v_i,v_j) \in E(G) \mid v_i \in S_{1}(t), v_j \in S_{1}(t), x_i(t+1) \neq x_j(t+1) \};$
\item $B'(t) = \{ (v_i,v_j) \in E(G) \mid v_i \in S_{2}(t), v_j \in S_{2}(t), x_i(t+1) \neq x_j(t+1) \};$
\item $C'(t) = \{ (v_i,v_j) \in E(G) \setminus (A'(t) \cup B'(t)) \mid x_i(t+1) \neq x_j(t+1)\}.$
\end{itemize}
In a manner similar to the above, we have
\begin{equation}
\sum\limits_{i \in S_{1}(t)}^{} \mathit{op}_{i}(t+1) - \sum\limits_{i \in S_{2}(t)}^{} \mathit{op}_{i}(t+1) = 2|A'(t)| - 2|B'(t)|.
\end{equation}

We finally note that $A(t) = A'(t)$ and $B(t) = B'(t)$, whence
\begin{equation}
\sum\limits_{i \in S_{1}(t)}^{} \mathit{op}_{i}(t+1) - \sum\limits_{i \in S_{2}(t)}^{} \mathit{op}_{i}(t+1) = 2|A(t)| - 2|B(t)|,
\end{equation}
which completes the proof.
\end{proof}

Let $\Delta E(t)$ be the variation in the energy function from time $t$ to time $t+1$, i.e.,
$\Delta E(t) = E(t+1) - E(t)$.
An important property of the energy function $E(t)$ is given in the following lemma.

\begin{lema}
\label{lema:crescente}
$E(t)$ is a monotonically nondecreasing function.
\end{lema}

\begin{proof}
By Lemma~\ref{lema:lema1},
\begin{equation}
\Delta E(t) = E(t+1) - E'(t).
\end{equation}
Then
\begin{equation}
\begin{split}
\Delta E(t) = \sum\limits_{i \in S_{1}(t+1)}^{} (\mathit{op}_{i}(t+1) - k) + \sum\limits_{i \in S_{2}(t+1)}^{} (k - \mathit{op}_{i}(t+1)) \\
 - \sum\limits_{i \in S_{1}(t)}^{} (\mathit{op}_{i}(t+1) - k) - \sum\limits_{i \in S_{2}(t)}^{} (k - \mathit{op}_{i}(t+1)).
\end{split}
\end{equation}

Focusing on each vertex's contribution to $\Delta E(t)$ yields:
\begin{itemize}
\item $\mathit{op}_{i}(t+1) - k - \mathit{op}_{i}(t+1) + k = 0$, if $v_i \in S_{1}(t)$ and $v_i \in S_{1}(t+1)$;
\item $k - \mathit{op}_{i}(t+1) - k + \mathit{op}_{i}(t+1) = 0$, if $v_i \in S_{2}(t)$ and $v_i \in S_{2}(t+1)$;
\item $k - \mathit{op}_{i}(t+1) - \mathit{op}_{i}(t+1) + k = 2(k - \mathit{op}_{i}(t+1)) > 0$, since
$k > \mathit{op}_{i}(t+1)$, if $v_i \in S_{1}(t)$ and $v_i \in S_{2}(t+1)$;
 \item $\mathit{op}_{i}(t+1) - k - k + \mathit{op}_{i}(t+1) = 2(\mathit{op}_{i}(t+1) - k) \geq 0$, since
$\mathit{op}_{i}(t+1) \geq k$, if $v_i \in S_{2}(t)$ and $v_i \in S_{1}(t+1)$.
\end{itemize}
The lemma follows from noting that
no vertex contributes negatively to $\Delta E(t)$.
\end{proof}

\section{\boldmath Application of the energy function to the study of the periodic behavior of $k$-reversible processes} \label{results}

As $G$ is finite, $E(t)$ cannot grow indefinitely and there exists a time $t_{\max}$ such that $E(t) = E(t_{\max})$ for all $t \geq t_{\max}$.

\begin{teorema}
\label{teo:periodo}
$p(x(0)) \leq 2$ for any $x(0)$.
\end{teorema}

\begin{proof}
Let $t_{0}$ be such 
that $t_0 > \tau(x(0))$ and $t_0 > t_{\max}$. Also, suppose that $p(x(0)) > 2$. For all $t > t_{0}$, there cannot be a vertex $v_{i}$ such that 
$v_{i} \in S_{1}(t)$ and 
$v_{i} \in S_{2}(t+1)$, since this would lead to $\Delta E(t) > 0$, contradicting $t > t_{\max}$ by Lemma~\ref{lema:crescente}.
Likewise, it cannot be the case that all vertices are in both $S_{1}(t)$ and $S_{1}(t+1)$, or in both $S_{2}(t)$ and $S_{2}(t+1)$, since these cases represent periodic behaviors of
length at most $2$, which contradicts our supposition. Hence, in order for inequality $p(x(0)) > 2$ to hold, there exists at least one vertex $v_{i}$ such that $v_{i} \in S_{2}(t)$ 
and $v_{i} \in S_{1}(t+1)$, and thus it holds that $x_i(t) \neq x_i(t+2)$. That is, there exists a vertex $v_{i}$ such that $\mathit{op}_{i}(t) < k$ and $\mathit{op}_{i}(t+1) = k$, 
otherwise the energy function would increase.

As $v_i \in S_{1}(t+1)$, then necessarily $v_{i} \in S_{1}(t+2)$, $v_{i} \in S_{1}(t+3)$, $v_{i} \in S_{1}(t+4)$, and so on. Otherwise, if $v_i \in S_{2}(t')$ for some $t' > t +1$, then
again the energy function would increase, contradicting the assumption that $t > t_{\max}$. Thus, $x_i(t) \neq x_i(t+2)$ and $x_i(t'+2) = x_i(t')$ for every
$t' > t$. Hence, $v_i$ reaches a periodic behavior only at time $t > t_{0}$, contradicting $t_{0} > \tau(x(0))$. It follows that only vertices with periodic behavior of length $2$ can exist,
thence the theorem.
\end{proof}

The transient length of a $k$-reversible process is closely related to $E(t_{\max})$, since $\Delta E(t) = 0$ for all $t \geq \tau(x(0))$, otherwise
it would be true that $x(t+2) \neq x(t)$. If $t < \tau(x(0))$, then $\Delta E(t) = 0$ only if there is at least one vertex $v_{i}$ such that $v_{i} \in S_2(t)$, $v_{i} \in
S_{1}(t+1)$, and $\mathit{op}_{i}(t+1) = k$. Thus, the energy function may remain unchanged for at most $n$ consecutive steps during the transient phase. Consider a configuration $x(t')$ and
the associated $E(t')$. This value of the energy function would be the same for exactly $n$ consecutive steps if we had $|S_1(t')| = 1$, $|S_2(t')| = n-1$, and the following 
conditions were true:
\begin{itemize}
\item $S_1(t') \varsubsetneq S_1(t'+1) \varsubsetneq \cdots  \varsubsetneq S_1(t'+n-1)$;
\item $S_2(t') \varsupsetneq S_2(t'+1) \varsupsetneq \cdots \varsupsetneq S_2(t'+n-1)$;
\item If $v_i \in S_2(j) \cap S_1(j+1)$, then $\mathit{op}_i(j+1) = k$.
\end{itemize}

\begin{teorema}
\label{teo:trans}
$\tau(x(0)) \leq E(t_{\max}) + n - 1$ for any $x(0)$.
\end{teorema}

\begin{proof}
In order for the energy function to remain unchanged within a time step, by the argument in the proof of Lemma~\ref{lema:crescente}
it is necessary that at least one vertex $v_i$ exists for which $v_i \in S_{2}(t)$, $v_i \in S_{1}(t+1)$,
and $\mathit{op}_i(t+1) = k$. We have two cases:
\begin{itemize}
\item $v_i$ belongs to the sets $S_2(0), S_2(1), \ldots, S_2(t)$.
\item There is a time $t'$ such that $0 < t' < t $ and $v_i \in S_1(t'), v_i \in S_2(t'+1), v_i \in S_2(t'+2), \ldots, v_i \in S_2(t)$.
\end{itemize}

Only in the second case does vertex $v_i$ contribute to increase the energy function from its value at some time prior to $t$,
since the transition from $S_1(t')$ to $S_2(t'+1)$ implies such an increase, again as in the proof of Lemma~\ref{lema:crescente}. Thus, except for the vertices
for which the first case holds, all other vertices that help to keep the energy function constant at time $t$ were responsible for increasing it at a previous time step by at least $2$.
So, apart from the first-case vertices, the increase in the energy function is, on average, of at least $1$ per time step. But
the first case can only hold for a vertex once, at time $t$. After $t$, any contribution from a vertex to keep the energy function constant is necessarily preceded by an increase of at least $2$.
Hence, the transient length can be at most $E(t_{\max}) + |S_{2}(0)|$. However, if $|S_{2}(0)| = n$ then $x(0)$ is already a periodic configuration. The theorem follows by using $|S_2(0)|<n$.
\end{proof}

\begin{corolario}
\label{cor:transLimite}
$\tau(x(0)) \leq n(\Delta(G) + 1) - 1$ for any $x(0)$.
\end{corolario}

\begin{proof}
We only need to consider the case of $k \leq \Delta(G)$. The maximum value of $E(t)$ is at most $n\Delta(G)$, when 
$k = \Delta(G)$ and $|S_2(t)| = n$. The corollary follows directly from Theorem~\ref{teo:trans}.
\end{proof}

For the cases in which $2k > \Delta(G)$, this bound can be improved.

\begin{corolario}
\label{cor:transLimiteMelhor}
If $2k > \Delta(G)$, then $\tau(x(0)) \leq n(k + 1) - 1$ for any $x(0)$.
\end{corolario}
 
\begin{proof}
The corollary follows directly from Theorem~\ref{teo:trans} by noting
that $E(t)\leq nk$ when $2k > \Delta(G)$.
\end{proof}

We proceed by studying the case of trees.

\begin{teorema}
If $G$ is a tree, then $E(t_{\max}) = nk$.
\label{teo:nk}
\end{teorema}

\begin{proof}
Using (\ref{eq:abreEt}) and (\ref{eq:atbt}) from the proof of Lemma~\ref{lema:lema1}, we obtain
\begin{equation}
E(t) = 2(|A(t)| - |B(t)|) + k(|S_2(t)| - |S_1(t)|).
\end{equation}

If $|S_2(t)| = n$ and $|S_1(t)| = 0$, then in order to maximize $E(t)$ we need to assume that all vertices in $S_2(t)$ have the same state at time $t$, and
consequently $|B(t)| = 0$. As $S_1(t)$ is empty, $|A(t)| = 0$ and
\begin{equation}
E(t) = 2(|A(t)| - |B(t)|) + k(|S_2(t)| - |S_1(t)|) = nk.
\end{equation}

Now assume that $|S_1(t)| = w$ and $|S_2(t)| = n-w$, with $0 < w < n$. When all vertices in $S_2(t)$ are in the same state at time $t$, yielding $|B(t)| = 0$, and moreover
$|A(t)| = |S_1(t)| - 1$, we have
\begin{equation}
E(t) \leq 2(|A(t)| - |B(t)|) + k(|S_2(t)| - |S_1(t)|) = kn + w(2-2k) - 2.
\end{equation}

Note that $E(t) > nk$ if and only if $w(2-2k) > 2$. However, for any $k \geq 1$, $w(2-2k) \leq 0$. Hence, whenever $w > 0$, there is no configuration $x(t)$ such that $E(t) > nk$. It follows
that the maximum value of the energy function in a tree occurs when all vertices have the same state, thence the theorem.
\end{proof}

\begin{corolario}
\label{cor:tree}
If $G$ is a tree, then $\tau(x(0)) \leq n(k+1) - 1$ for any $x(0)$.
\end{corolario}

\begin{proof}
The corollary follows directly from Theorems~\ref{teo:trans} and~\ref{teo:nk}.
\end{proof}

Corollary~\ref{cor:tree} is a clear improvement on Corollary~\ref{cor:transLimite}. As $k$ is a constant, the upper bound
on the length of the transient is linear in the number of vertices.

\section{Conclusions}\label{concl}

We have presented an energy function that is useful to study the periodic behavior of $k$-reversible processes. We have shown that this function is monotonic, and using 
this fact we have presented and alternative proof that the maximum period length is $2$ for such processes, and also that the transient length is $O(n\Delta(G))$ in the general case and $O(nk)$
in the case of trees.

For $2$-reversible processes on trees, we conjecture that the sharp upper bound on the transient length is $n-3$, provided $n \geq 5$.
We give support to this conjecture in Appendix~\ref{app:suppl}.

\section*{Acknowledgments}

The authors acknowledge partial support from CNPq, CAPES, and FAPERJ BBP grants.

\bibliographystyle{plain}
\bibliography{efunction}

\appendix
\section{Supplementary material}\label{app:suppl}
\input{efunction-suppl}

\end{document}

%% file: efunction-suppl.tex
\begin{conjectura}
For $2$-reversible processes on trees with $n \geq 5$, the sharp upper bound on the transient length is $n-3$.
\end{conjectura}

We have confirmed this conjecture computationally by exhaustive enumeration for $n\leq 20$.
Throughout these experiments, a pattern emerged regarding the trees and initial configurations for which a transient length of exactly $n-3$ was verified.
This pattern is such that the upper bound is achieved by exactly $\frac{n}{2}$ trees for $n$ even and exactly $\frac{n-1}{2} - 1$ trees for $n$ odd, always for
exactly one initial configuration of the tree in question.
Such trees and configurations can be generated by Algorithm~\ref{alg:conjectura}, whose results are shown in
Figure~\ref{fig:arvs1} for $n=8$ and in Figure~\ref{fig:arvs2} for $n = 9$.

\begin{algorithm}[p]
    \SetKwInput{Input}{Input}\SetKwInput{Output}{Output}
    \label{alg:conjectura}
    \caption{Generate all trees and corresponding initial configurations leading to a transient length of $n-3$ in a $2$-reversible process.}
    \Input{Number $n$ of vertices.}
    \Output{Trees $T_j$, each with the corresponding initial configuration $X_j$ such that $\tau(X_j) = n-3$.}
    \BlankLine
    \Begin{
        $V \leftarrow \{v_1, v_2, \dots, v_n\}$\;
        $E \leftarrow \emptyset$\;
	$j \leftarrow 1$\;
        \For{$i \leftarrow 1$ to $n-2$}
	{
	  $E \leftarrow E \cup (v_i,v_{i+1})$\;
	  \If{$i = n-2$}
	  {
	    $E \leftarrow E \cup (v_{i},v_{i+2})$\;
	  }
	}

	\For{$i \leftarrow 1$ to $n$}
	{
	  \If{$i$ is odd}
	  {
	    $x[i] \leftarrow +1$\;
	  }
	  \Else
	  {
	    $x[i] \leftarrow -1$\;
	  }
	}
	$T_j \leftarrow G(V,E)$\;
	$X_j \leftarrow x$\;
	$j \leftarrow j + 1$\;
        \For{$i \leftarrow 3$ to $n-3$}
	{
	  \If{$i$ is odd}
	  {
	    $E \leftarrow E \setminus (v_i,v_{i-1})$\;
	    $E \leftarrow E \cup (v_{i-1}, v_{i+1})$\;
	    $T_j \leftarrow G(V,E)$\;
	    $X_j \leftarrow x$\;
	    $j \leftarrow j + 1$\;
	  }
	  \If{$i = n-3$ and $n$ is even}
	  {
	    $E \leftarrow E \setminus (v_{i+1}, v_{i+3})$\;
	    $E \leftarrow E \cup (v_i, v_{i+3})$\;
	    $T_j \leftarrow G(V,E)$\;
	    $X_j \leftarrow x$\;
	    $j \leftarrow j + 1$\;
	  }
	}
    }
\end{algorithm}

\begin{figure}[p]
 \centering
\begin{tikzpicture}
  [scale=.7,auto=left, inner sep=2pt]
  \node[draw,fill=lightgray,circle] (n1)  at (0,0) {$v_{1}$};
  \node[draw,fill=white,circle] (n2)  at (1,0) {$v_{2}$};
  \node[draw,fill=lightgray,circle] (n3)  at (2,0) {$v_{3}$};
  \node[draw,fill=white,circle] (n4)  at (3,0) {$v_{4}$};
  \node[draw,fill=lightgray,circle] (n5)  at (4,0) {$v_{5}$};
  \node[draw,fill=white,circle] (n6)  at (5,0) {$v_{6}$};
  \node[draw,fill=lightgray,circle] (n7)  at (5,-1) {$v_{7}$};  
  \node[draw,fill=white,circle] (n8)  at (6,0) {$v_{8}$}; 
  \foreach \from/\to in {n1/n2, n2/n3, n3/n4, n4/n5, n5/n6, n6/n7, n6/n8}
  \draw (\from) -- (\to);
\end{tikzpicture}
\\
\vspace{2em}
\begin{tikzpicture}
  [scale=.7,auto=left,inner sep=2pt]
  \node[draw,fill=lightgray,circle] (n1)  at (1,0) {$v_{1}$};
  \node[draw,fill=white,circle] (n2)  at (2,0) {$v_{2}$};
  \node[draw,fill=lightgray,circle] (n3)  at (3,-1) {$v_{3}$};
  \node[draw,fill=white,circle] (n4)  at (3,0) {$v_{4}$};
  \node[draw,fill=lightgray,circle] (n5)  at (4,0) {$v_{5}$};
  \node[draw,fill=white,circle] (n6)  at (5,0) {$v_{6}$};
  \node[draw,fill=lightgray,circle] (n7)  at (5,-1) {$v_{7}$};  
  \node[draw,fill=white,circle] (n8)  at (6,0) {$v_{8}$}; 
  \foreach \from/\to in {n1/n2, n2/n4, n3/n4, n4/n5, n5/n6, n6/n7, n6/n8}
  \draw (\from) -- (\to);
\end{tikzpicture}
\\
\vspace{2em}
\begin{tikzpicture}
  [scale=.7,auto=left,inner sep=2pt]
  \node[draw,fill=lightgray,circle] (n1)  at (2,0) {$v_{1}$};
  \node[draw,fill=white,circle] (n2)  at (3,0) {$v_{2}$};
  \node[draw,fill=lightgray,circle] (n3)  at (4,-1) {$v_{3}$};
  \node[draw,fill=white,circle] (n4)  at (4,0) {$v_{4}$};
  \node[draw,fill=lightgray,circle] (n5)  at (5,1) {$v_{5}$};
  \node[draw,fill=white,circle] (n6)  at (5,0) {$v_{6}$};
  \node[draw,fill=lightgray,circle] (n7)  at (5,-1) {$v_{7}$};  
  \node[draw,fill=white,circle] (n8)  at (6,0) {$v_{8}$}; 
  \foreach \from/\to in {n1/n2, n2/n4, n3/n4, n4/n6, n5/n6, n6/n7, n6/n8}
  \draw (\from) -- (\to);
\end{tikzpicture}
\\
\vspace{2em}
\begin{tikzpicture}
  [scale=.7,auto=left,inner sep=2pt]
  \node[draw,fill=lightgray,circle] (n1)  at (2,0) {$v_{1}$};
  \node[draw,fill=white,circle] (n2)  at (3,0) {$v_{2}$};
  \node[draw,fill=lightgray,circle] (n3)  at (4,-1) {$v_{3}$};
  \node[draw,fill=white,circle] (n4)  at (4,0) {$v_{4}$};
  \node[draw,fill=lightgray,circle] (n5)  at (6,0) {$v_{5}$};
  \node[draw,fill=white,circle] (n6)  at (5,0) {$v_{6}$};
  \node[draw,fill=lightgray,circle] (n7)  at (5,-1) {$v_{7}$};  
  \node[draw,fill=white,circle] (n8)  at (7,0) {$v_{8}$}; 
  \foreach \from/\to in {n1/n2, n2/n4, n3/n4, n4/n6, n5/n6, n5/n8, n6/n7}
  \draw (\from) -- (\to);
\end{tikzpicture}
\caption{Output of Algorithm~\ref{alg:conjectura} for $n = 8$ ($T_1$ through $T_4$, from top to bottom). Shaded circles
indicate state $+1$, empty circles indicate state $-1$.}
\label{fig:arvs1}
\vspace{0.5in}
\begin{tikzpicture}
  [scale=.7,auto=left,inner sep=2pt]
  \node[draw,fill=lightgray,circle] (n1)  at (0,0) {$v_{1}$};
  \node[draw,fill=white,circle] (n2)  at (1,0) {$v_{2}$};
  \node[draw,fill=lightgray,circle] (n3)  at (2,0) {$v_{3}$};
  \node[draw,fill=white,circle] (n4)  at (3,0) {$v_{4}$};
  \node[draw,fill=lightgray,circle] (n5)  at (4,0) {$v_{5}$};
  \node[draw,fill=white,circle] (n6)  at (5,0) {$v_{6}$};
  \node[draw,fill=lightgray,circle] (n7)  at (6,0) {$v_{7}$};  
  \node[draw,fill=white,circle] (n8)  at (7,0) {$v_{8}$};
  \node[draw,fill=lightgray,circle] (n9)  at (6,-1) {$v_{9}$};  
  \foreach \from/\to in {n1/n2, n2/n3, n3/n4, n4/n5, n5/n6, n6/n7, n7/n8, n7/n9}
  \draw (\from) -- (\to);
\end{tikzpicture}
\\
\vspace{2em}
\begin{tikzpicture}
  [scale=.7,auto=left,inner sep=2pt]
  \node[draw,fill=lightgray,circle] (n1)  at (1,0) {$v_{1}$};
  \node[draw,fill=white,circle] (n2)  at (2,0) {$v_{2}$};
  \node[draw,fill=lightgray,circle] (n3)  at (3,-1) {$v_{3}$};
  \node[draw,fill=white,circle] (n4)  at (3,0) {$v_{4}$};
  \node[draw,fill=lightgray,circle] (n5)  at (4,0) {$v_{5}$};
  \node[draw,fill=white,circle] (n6)  at (5,0) {$v_{6}$};
  \node[draw,fill=lightgray,circle] (n7)  at (6,0) {$v_{7}$};  
  \node[draw,fill=white,circle] (n8)  at (7,0) {$v_{8}$};
  \node[draw,fill=lightgray,circle] (n9)  at (6,-1) {$v_{9}$};  
  \foreach \from/\to in {n1/n2, n2/n4, n3/n4, n4/n5, n5/n6, n6/n7, n7/n8, n7/n9}
  \draw (\from) -- (\to);
\end{tikzpicture}
\\
\vspace{2em}
\begin{tikzpicture}
  [scale=.7,auto=left,inner sep=2pt]
  \node[draw,fill=lightgray,circle] (n1)  at (2,0) {$v_{1}$};
  \node[draw,fill=white,circle] (n2)  at (3,0) {$v_{2}$};
  \node[draw,fill=lightgray,circle] (n3)  at (4,-1) {$v_{3}$};
  \node[draw,fill=white,circle] (n4)  at (4,0) {$v_{4}$};
  \node[draw,fill=lightgray,circle] (n5)  at (5,-1) {$v_{5}$};
  \node[draw,fill=white,circle] (n6)  at (5,0) {$v_{6}$};
  \node[draw,fill=lightgray,circle] (n7)  at (6,0) {$v_{7}$};  
  \node[draw,fill=white,circle] (n8)  at (7,0) {$v_{8}$};
  \node[draw,fill=lightgray,circle] (n9)  at (6,-1) {$v_{9}$};  
  \foreach \from/\to in {n1/n2, n2/n4, n3/n4, n4/n6, n5/n6, n6/n7, n7/n8, n7/n9}
  \draw (\from) -- (\to);
\end{tikzpicture}
\caption{Output of Algorithm~\ref{alg:conjectura} for $n = 9$ ($T_1$ through $T_3$, from top to bottom). Shaded circles
indicate state $+1$, empty circles indicate state $-1$.}
\label{fig:arvs2}
\end{figure}

%% file: efunction.bbl
\begin{thebibliography}{10}

\bibitem{bootstrap-example}
J.~Adler.
\newblock Bootstrap percolation.
\newblock {\em Physica A}, 171:453--470, 1991.

\bibitem{opinion}
M.~H. De{G}root.
\newblock Reaching a consensus.
\newblock {\em J. Am. Stat. Assoc.}, 69:167--182, 1974.

\bibitem{processosReversiveis}
P.~A. {Dreyer~Junior}.
\newblock {\em Application and Variations of Domination in Graphs}.
\newblock {Ph.D.} dissertation, The State University of New Jersey, New
  Brunswick, NJ, 2000.

\bibitem{phase-wrapping}
D.~C. Ghiglia, G.~A. Mastin, and L.~A. Romero.
\newblock Cellular automata method for phase unwrapping.
\newblock {\em J. Opt. Soc. Am. A}, 4:267--280, 1987.

\bibitem{golesOlivos2}
E.~Goles and J.~Olivos.
\newblock The convergence of symmetric threshold automata.
\newblock {\em Inform. Control}, 51:98--104, 1981.

\bibitem{golesOlivos}
E.~Goles and J.~Olivos.
\newblock Periodic behavior of binary threshold functions and applications.
\newblock {\em Discrete Appl. Math.}, 3:93--105, 1985.

\bibitem{energyFunction}
E.~Goles-Chacc, F.~Fogelman-Soulie, and D.~Pellegrin.
\newblock Decreasing energy functions as a tool for studying threshold
  networks.
\newblock {\em Discrete Appl. Math.}, 12:261--277, 1985.

\bibitem{hopfield}
J.~J. Hopfield.
\newblock Neural networks and physical systems with emergent collective
  computational abilities.
\newblock {\em Proc. Natl. Acad. Sci. USA}, 79:2554--2558, 1987.

\bibitem{disease-example}
A.~R. Mikler, S.~Venkatachalam, and K.~Abbas.
\newblock Modeling infectious diseases using global stochastic cellular
  automata.
\newblock {\em J. Biol. Syst.}, 13:421--439, 2005.

\bibitem{poljak}
S.~Poljak and M.~S\r{u}ra.
\newblock On periodical behaviour in societies with symmetric influences.
\newblock {\em Combinatorica}, 3:119--121, 1983.

\end{thebibliography}
